\documentclass[reqno,fleqn,12pt]{amsart}
\usepackage{amsmath}
\usepackage{latexsym}
\usepackage{amssymb}
\usepackage{graphics}

\usepackage{tikz}
\usetikzlibrary{calc}
\usetikzlibrary{positioning}

\newtheorem{theorem}{Theorem}     
\newtheorem{corollary}{Corollary} 
\newtheorem{lemma}{Lemma} 

\newtheorem{remark}{Remark} 
\newtheorem{definition}{Definition} 

\def\rbm{\textsc{rainbow matching}}

\def\mrbm{\textsc{max rainbow matching}}

\newcommand{\eps}{\varepsilon}

\title{Complexity Results for Rainbow Matchings}
\author[V.B. Le]{Van Bang Le}
\curraddr[VBL]{Universit\"at Rostock\\ Institut f\"ur Informatik\\ 18051 Rostock, Germany}
\email{le@informatik.uni-rostock.de}
\author[F. Pfender]{Florian Pfender}
\curraddr[FP]{University of Colorado at Denver, Department of Mathematics \& Statistics,\\  
  Denver, CO 80202, USA}
\email{florian.pfender@ucdenver.edu}
\subjclass {}

\begin{document}

\begin{abstract}
A rainbow matching in an edge-colored graph is a matching whose edges have
distinct colors. We address the complexity issue of the following problem, 
\mrbm: Given an edge-colored graph $G$, how large is the largest rainbow 
matching in $G$? We present several sharp contrasts in the complexity of 
this problem. 

We show, among others, that 
\begin{itemize}
 \item \mrbm\ can be approximated by a polynomial algorithm with approximation ratio 
       $2/3-\eps$. 
 \item \mrbm\ is APX-complete, even when restricted to properly edge-colored 
    linear forests without a $5$-vertex path, and is solvable in 
    polynomial time for edge-colored forests without a $4$-vertex path.
 \item \mrbm\ is APX-complete, even when restricted to properly edge-colored 
    trees without an $8$-vertex path, and is solvable in 
    polynomial time for edge-colored trees without a $7$-vertex path.
 \item \mrbm\ is APX-complete, even when restricted to properly edge-colored paths.    
\end{itemize}
These results provide a dichotomy theorem for the complexity of the 
problem on forests and trees in terms of forbidding paths. 
The latter is somewhat surprising, since, to the best of our knowledge, 
no (unweighted) graph problem prior to our result is known to be NP-hard 
for simple paths.  

We also address the parameterized complexity of the problem.
\end{abstract}

\maketitle

\section{Introduction and Results}

Given a graph $G=(E(G),V(G))$, an edge coloring is a function 
$\phi : E(G)\rightarrow \mathcal{C}$ mapping each edge 
$e\in E(G)$ to a color $\phi(e)\in \mathcal{C}$; $\phi$ is a \emph{proper} edge-coloring 
if, for all distinct edges $e$ and $e'$, $\phi(e)\not=\phi(e')$ whenever $e$ and $e'$ have 
an endvertex in common. A (properly) edge-colored graph $(G,\phi)$ is a pair of a graph 
together with a (proper) edge coloring. A {\em rainbow subgraph} of an edge-colored graph 
is a subgraph whose edges have distinct colors. Rainbow subgraphs appear frequently in the 
literature, for a recent survey we point to~\cite{KL}.

In this paper we are concerned with \emph{rainbow matchings}, i.e., matchings whose edges 
have distinct colors. One motivation to look at rainbow matchings is Ryser's famous conjecture 
from~\cite{Ry}, which states that every Latin square of odd order contains a Latin 
transversal. Equivalently, the conjecture says that every proper edge coloring of the 
complete bipartite graph $K_{2n+1,2n+1}$ with $2n+1$ colors contains a rainbow matching 
with $2n+1$ edges. 

One often asks for the size of the largest rainbow matching in an edge-colored graph with 
certain restrictions (see, e.g., \cite{DFLMPW,KY,LT,Wa}). In this paper, we consider the 
complexity of this problem.
In particular, we consider the complexity of the following two problems, and we will
restrict them to certain graph classes and edge colorings.

\medskip
\noindent
\rbm\\[1ex]
\begin{tabular}{ll}
Instance: & Graph $G$ with an edge-coloring and an integer $k$\\
Question: & Does $G$ have a rainbow matching with at least $k$ edges?\\
\end{tabular}

\medskip
\noindent
\rbm\ is also called \textsc{multiple choice matching} in \cite[Problem GT55]{GJ}. 
The optimization version of the decision problem \rbm\ is:

\medskip
\noindent
\mrbm\\[1ex]
\begin{tabular}{ll}
Instance: & Graph $G$ with an edge-coloring\\
Output:   & A largest rainbow matching in $G$.
\end{tabular}

\medskip
\noindent
Only the following complexity result for \rbm\ is known. 
Note that the considered graphs in the proof are not properly edge colored. 
When restricted to properly edge-colored graphs, no complexity result is known prior to this 
work.

\begin{theorem}[\cite{IRT}]\label{thm0}
\rbm\ is NP-complete, even when restricted to edge-colored bipartite graphs.
\end{theorem}

In this paper, we analyze classes of graphs for which \mrbm\ can be solved and 
thus \rbm\ can be decided in polynomial time, or is NP-hard. Our results are:

\begin{itemize}
 \item There is a polynomial-time $(2/3-\eps)$-approximation algorithm for \mrbm\ 
    for every $\eps>0$.
 \item \mrbm\ is APX-complete, and thus \rbm\ is NP-complete, even for very 
    restricted graphs classes such as
    \begin{itemize}
      \item edge-colored complete graphs,
      \item properly edge-colored paths, 
      \item properly edge-colored $P_8$-free trees in which every color is used at 
         most twice, 
      \item properly edge-colored $P_5$-free linear forests in which every color 
         is used at most twice,
      \item properly edge-colored $P_4$-free bipartite graphs in which every color 
         is used at most twice.
    \end{itemize}
    These results significantly improve Theorem~\ref{thm0}. 
    We also provide an inapproximability bound for each of the listed graph classes. 
  \item \mrbm\ is solvable in time $O(m^{3/2})$ for $m$-edge graphs without $P_4$ 
     (induced or not); in particular for $P_4$-free forests.
  \item \mrbm\ is polynomially solvable for $P_7$-free forests with bounded number 
     of components; in particular for $P_7$-free trees.
  \item \mrbm\ is fixed parameter tractable for $P_5$-free forests, when parameterized 
     by the number of the components.
\end{itemize}

The next section contains some relevant notation and definitions. 
Section~\ref{sec:approx-hard} deals with approximability and inapproximability results, 
Section~\ref{sec:poly} discusses some polynomially solvable cases, and Section~\ref{sec:fixed} 
addresses the parameterized complexity. 
We conclude the paper in Section~\ref{sec:conclusion} with some open problems.

\section{Definitions and Preliminaries}\label{sec:defs}

We consider only finite, simple, and undirected graphs.
For a graph $G$, the vertex set is denoted $V(G)$ and the edge set is denoted $E(G)$. 
An edge $xy$ of a graph $G$ is a {\it bridge} if $G-xy$ has more components 
than $G$.
If $G$ does not contain an induced subgraph isomorphic to another graph $F$, 
then $G$ is {\it $F$-free}. 

For $\ell \ge 1$, let $P_\ell$ denote a chordless path with $\ell$ vertices and $\ell-1$ 
edges, and for $\ell \ge 3$, let $C_\ell$ denote a chordless cycle with $\ell$ vertices 
and $\ell$ edges. A {\it triangle} is a $C_3$. 
For $p, q\ge 1$, $K_{p,q}$ denotes the complete bipartite graph with $p$ vertices of one 
color class and $q$ vertices of the second color class; a {\it star} is a $K_{1,q}$. 
A complete graph with $p$ vertices is denoted by $K_p$; $K_p-e$ is the graph obtained from 
$K_p$ by deleting one edge.
An {\it $r$-regular graph} is one in which each vertex has degree exactly $r$.
A forest in which each component is a path is a {\it linear forest}. 

The {\it line graph} $L(G)$ of a graph $G$ has vertex set $E(G)$, and 
two vertices in $L(G)$ are adjacent if the corresponding edges in $G$ are  
incident. By definition, every matching in $G$ corresponds to an independent set in $L(G)$ 
of the same size, and vice versa. One of the main tools we use in discussing rainbow matchings 
is the following concept that generalizes line graphs naturally: 

\begin{definition}
The {\em color-line graph} $C\!L(G)$ of an edge-colored graph $G$ has vertex set $E(G)$, and 
two vertices in $C\!L(G)$ are adjacent if the corresponding edges in $G$ are 
incident or have the same color.
\end{definition}

Notice that, given an edge-colored graph $G$, $C\!L(G)$ can be constructed in time 
$O(|E(G)|^2)$ in an obvious way. We will make use of further facts about color-line graphs 
below that can be verified by definition.

\begin{lemma}\label{lemma:CL}
Let $G$ be an edge-colored graph. Then
\begin{itemize}
 \item[(i)]  $C\!L(G)$ is $K_{1,4}$-free.
 \item[(ii)] $C\!L(G)$ is $(K_7-e)$-free, provided $G$ is properly edge-colored.
 \item[(iii)] Every rainbow matching in $G$ corresponds to an independent set in $C\!L(G)$ 
    of the same size, and vice versa.
\end{itemize} 
\end{lemma}

Lemma~\ref{lemma:CL} allows us to use results about independent 
sets to obtain results on rainbow matchings. 
This way, we will relate \mrbm\ to the following two problems, 
which are very well studied in the literature.

\medskip
\noindent
\textsc{MIS (Maximum Independent Set)}\\[1ex] 
\begin{tabular}{lp{0.8\textwidth}}
Instance: & A graph $G$.\\
Output: & A maximum independent set in $G$.
\end{tabular}

\medskip
\noindent
\textsc{$3$-MIS (Maximum Independent Set in $3$-regular Graphs)}\\[1ex] 
\begin{tabular}{lp{0.8\textwidth}}
Instance: & A $3$-regular graph $G$.\\
Output: & A maximum independent set in $G$.
\end{tabular}

\medskip
\noindent
In the present paper, a polynomial-time algorithm $\mathrm{A}$ with approximation ratio 
$\alpha$, $0<\alpha<1$, for a (maximization) problem is one that, for all problem instances 
$I$, $\mathrm{A}(I)\ge \alpha\cdot\mathrm{opt}(I)$, where $\mathrm{A}(I)$ is the objective
value of the solution found by $\mathrm{A}$ and $\mathrm{opt}(I)$ 
is the objective value of an optimal solution. 
A problem is said to be in APX (for approximable) if it admits an algorithm with a 
constant approximation ratio. A problem in APX is called APX-complete if all other problems 
in APX can be $L$-reduced (cf.~\cite{PY}) to it. 
It is known that \textsc{$3$-MIS} is APX-complete (see~\cite{AK}). Thus, if \textsc{$3$-MIS} 
is $L$-reducible to a problem in APX, then this problem is also APX-complete. 
All reductions in this paper are $L$-reductions.

\section{Approximability and Hardness}\label{sec:approx-hard}

We first show that \mrbm\ is in APX, by reducing to {\sc MIS} on $K_{1,4}$-free graphs. 
The following theorem is due to Hurkens and Schrijver~\cite{HS}; see also~\cite{Hal}.

\begin{theorem}[\cite{HS}]\label{thm:HS}
For every $\eps>0$ and $p\ge 3$, {\sc MIS} for $K_{1,p+1}$-free graphs can be approximated by 
a polynomial algorithm with approximation ratio $2/p-\eps$.
\end{theorem}

\begin{theorem}\label{thm:apx}
For every $\eps>0$, \mrbm\ can be approximated by a polynomial algorithm with approximation 
ratio $2/3-\eps$. 
\end{theorem}

\begin{proof}
This follows from Lemma~\ref{lemma:CL} and Theorem~\ref{thm:HS} with $p=3$.
\end{proof}

On the other hand, we show that \mrbm\ is APX-complete, and thus \rbm\ is NP-complete, even 
when restricted to very simple graph classes, and we give some inapproximability bounds 
for \mrbm. 
We will reduce \mrbm\ on these graph classes to {\sc $3$-MIS}, and use the following 
theorem by Berman and Karpinski~\cite{BK}, where the second part of the statement does not 
appear in the original statement, but follows directly from their proof.

\begin{theorem}[\cite{BK}]\label{thm:BK}
For any $\eps\in (0,1/2)$, it is NP-hard to decide whether an instance
of {\sc $3$-MIS} with $284n$ nodes has the
maximum size of an independent set above $(140-\eps)n$ or below $(139+\eps)n$. 
The statement remains true if we restrict ourselves to the class of bridgeless 
triangle-free $3$-regular graphs.
\end{theorem}

\begin{theorem}\label{thm:2fac}
\mrbm\ is APX-complete, even when restricted to properly edge-colored $2$-regular graphs in 
which every color is used exactly twice. Unless P=NP, no polynomial algorithm can guarantee 
an approximation ratio greater than $\frac{139}{140}$.
\end{theorem}
\begin{proof}
 Let $G$ be a bridgeless triangle-free $3$-regular graph on $284n$ vertices. 
By a classical theorem of Petersen~\cite{Pet}, $G$ contains a perfect matching $M$. 
Then, $G-M$ is triangle-free and $2$-regular, and thus the line graph of a 
triangle-free and $2$-regular graph $H$ on $284n$ vertices. Now it is easy to color 
the edges of $H$ in such a way that every color is used exactly twice, and $G=C\!L(H)$. 
From Theorem~\ref{thm:BK}, it follows that it is NP-hard to decide if the maximal 
size of a rainbow matching in $H$ is above $(140-\eps)n$ or below $(139+\eps)n$.
\end{proof}

\begin{corollary}\label{thm:apxcom}
\mrbm\ is APX-complete, even when restricted to complete graphs. Unless P=NP, no polynomial 
algorithm can guarantee an approximation ratio greater than $\frac{139}{140}$.
\end{corollary}
\begin{proof}
 Use the same graph $H$ from the previous proof, add two new vertices, and add all 
missing edges, all colored with the same new color, to get an edge-colored complete 
graph $H'$ on $284n+2$ vertices. Then, it is NP-hard to decide if the maximal size 
of a rainbow matching in $H'$ is above $(140-\eps)n+1$ or below $(139+\eps)n+1$.
\end{proof}

\begin{theorem}\label{thm:apxpath}
\mrbm\ is APX-complete, even when restricted to properly edge-colored paths. 
Unless P=NP, no polynomial algorithm can guarantee an approximation ratio greater than 
$\frac{210}{211}$.
\end{theorem}
\begin{proof}
 Again, start with the $2$-regular graph $H$ from the proof of Theorem~\ref{thm:2fac}. 
For a cycle $C\subseteq H$, create a path $P_C$ as follows. Cut the cycle at a vertex 
$v$ to get a path of the same length as $C$, with the two end vertices corresponding to the 
original vertex $v$. Now add an extra edge to each of the two ends, and color this 
edge with the color $v$---a color not used anywhere else in $H$. The maximum rainbow 
matching in this new graph $H'$ is exactly one greater than the maximum rainbow 
matching in $H$. To see this, take a maximum rainbow matching in $H$, and notice that 
it can contain at most one edge incident to $v$. Thus, in $H'$ we can add one of 
the two edges colored $v$ to this matching to get a greater rainbow matching. On the 
other hand, every rainbow matching in $H$ contains at most two edges incident to the 
two copies of $v$, and at most one of them is not in $H$. Thus, deleting one edge 
from a rainbow matching in $H'$ yields a rainbow matching in $H$.

Now repeat this process for every cycle in $H$ to get a linear forest $L$ with $c$ 
components, say, and $284n+2c$ edges. Similarly to above, it is NP-hard to decide 
if the maximal size of a rainbow matching in $L$ is above $(140-\eps)n+c$ or below 
$(139+\eps)n+c$. We now connect all paths in $L$ with $c-1$ extra edges colored with 
a new color $1$ to one long path, and add a path on $5$ edges colored $1,2,1,2,1$ (where $2$ is a new color) to 
one end to get a path $P$ on $284n+3c+4$ edges. The size of a maximum rainbow matching 
in $P$ is exactly $2$ larger than in $L$, so it is NP-hard to decide if the maximal 
size of a rainbow matching in $P$ is above $(140-\eps)n+c+2$ or below $(139+\eps)n+c+2$. 
As $H$ does not contain any triangles, we have $c\le 71n$, and the theorem follows.
\end{proof}

\begin{theorem}\label{thm:apxlf}
\mrbm\ is APX-complete, even when restricted to properly edge-colored $P_5$-free linear 
forests in which every color is used at most twice. Unless P=NP, no polynomial algorithm 
can guarantee an approximation ratio greater than $\frac{423}{424}$.
\end{theorem}
\begin{proof}
 We again start with the $2$-regular graph $H$ from the proof of Theorem~\ref{thm:2fac}. 
Now construct a linear forest $L$ consisting of $284n=|E(H)|$ paths of length $3$, 
where every edge in $H$ corresponds to one path component in $L$. For an edge 
$vw\in E(H)$ with color $\phi(vw)$, color the three edges of the corresponding path 
with the colors $v$, $\phi(vw)$ and $w$ in this order. We claim that a maximum rainbow 
matching in $L$ is exactly $284n$ greater than a maximum rainbow matching in $H$. 
Note that this claim implies the theorem. To see the claim, consider first a rainbow 
matching $M$ in $H$. Note that for every vertex $v\in V(H)$, $M$ can contain at most 
one edge incident to $v$, so in $L$ we can add one of the two edges labeled $v$ to 
the matching induced by $M$. This can be done for every vertex in $V(H)$, so the 
largest rainbow matching in $L$ is at least $284n$ larger than $M$. On the other 
hand, every rainbow matching $M'$ in $L$ contains at most two edges either colored 
$v$ or incident to an edge colored $v$, and at most one of them is colored $v$. 
Thus, by deleting at most $284n$ edges from $M'$ we can create a rainbow matching 
in $H$. 
Therefore, it is NP-hard to decide if the maximal size of a rainbow matching in 
$L$ is above $(140+284-\eps)n$ or below $(139+284+\eps)n$.
\end{proof}

\begin{theorem}\label{thm:apxp4}
\mrbm\ is APX-complete, even when restricted to properly edge-colored $P_4$-free bipartite 
graphs in which every color is used at most twice. Unless P=NP, no polynomial algorithm can 
guarantee an approximation ratio greater than $\frac{423}{424}$.
\end{theorem}
\begin{proof}
 Take the linear forest $L$ from the proof of Theorem~\ref{thm:apxlf}, add an edge 
to every $P_4$ to make it a $C_4$, and color the new edge with the same color as 
the middle edge of the $P_4$. This graph $G$ is $P_4$-free, and every rainbow 
matching in $G$ corresponds to a rainbow matching in $L$ of the same size.
\end{proof}

\begin{theorem}\label{thm:apxlf6}
\mrbm\ is APX-complete, even when restricted to properly edge-colored $P_6$-free linear 
forests in which every color is used at most twice. Unless P=NP, no polynomial algorithm 
can guarantee an approximation ratio greater than $\frac{1689}{1694}$.
\end{theorem}
\begin{proof}
Similarly to above, we start with the $2$-regular graph $H$ from the proof of 
Theorem~\ref{thm:2fac}, and transform it into a linear forest $L$ similarly to 
the last two proofs. This time, we try to produce paths of length $4$ by splitting 
the cycles in $H$ only at every other vertex if possible. As $H$ may contain odd 
cycles, we have to use one path of length only $3$ for every one of the $o$ odd 
cycles in $H$. Thus, $L$ has exactly $(284n+o)/2$ component paths, and a maximum 
rainbow matching in $L$ is exactly $(284n+o)/2$ greater than a maximum rainbow 
matching in $H$. As $H$ is triangle-free, we know that $o\le 284n/5$, so it is 
NP-hard to decide if the maximal size of a rainbow matching in $L$ is above 
$(140+0.7\times 284-\eps)n$ or below $(139+0.7\times 284+\eps)n$.
\end{proof}

\begin{theorem}\label{thm:apxtree}
\mrbm\ is APX-complete, even when restricted to properly edge-colored $P_8$-free trees 
in which every color is used at most twice. Unless P=NP, no polynomial algorithm can 
guarantee an approximation ratio greater than $\frac{1689}{1694}$.
\end{theorem}
\begin{proof}
 Start with an edge-colored linear forest $L$ as in the proof of 
Theorem~\ref{thm:apxlf6}. Add an extra vertex $v$, and connect it to a central 
vertex (i.e., a vertex with maximum distance to the ends) in every path in $L$. 
Further, add one pending edge to $v$. Color the added edges with colors not 
appearing on $L$. Then the resulting tree $T$ is $P_8$-free, and a maximum 
rainbow matching in $T$ is exactly one edge larger than a maximum rainbow 
matching in $L$.
\end{proof}

All these theorems imply the following complexity result on \rbm.
\begin{corollary}
 \rbm\ is NP-complete, even when restricted to one of the following classes of 
edge-colored graphs.
\begin{enumerate}
 \item Complete graphs.
 \item Properly edge-colored paths.
\item Properly edge-colored $P_5$-free linear forests in which every color is used at 
      most twice.
\item Properly edge-colored $P_4$-free bipartite graphs in which every color is used at 
      most twice.
\item Properly edge-colored $P_8$-free trees in which every color is used at most twice.
\end{enumerate}
\end{corollary}

\section{Polynomial-time Solvable Cases}\label{sec:poly}

In contrast to Theorem~\ref{thm:apxp4}, saying that \mrbm\ is hard even for $P_4$-free 
bipartite graphs, we have:

\begin{theorem}\label{thm:sfor}
In every graph $G$ which does not contain $P_4$ as a not necessarily induced subgraph, 
\mrbm\ is solvable in time $O(m^{3/2})$, where $m$ is the number of edges in $G$.
\end{theorem}
\begin{proof}
As $G$ does not contain a $P_4$, every component of $G$ is either a star or a triangle. 
Now construct a bipartite graph $H$ with partite sets being the components of $G$ and 
the colors used in $G$. The graph $H$ has an edge between a component and a color if 
in $G$, the color appears in the component. The graph $H$ has at most as many edges 
as $G$, and rainbow matchings in $G$ correspond to matchings in $H$ of the same size. 
As we can find a maximum matching in the bipartite graph in time $O(m^{3/2})$ 
(\cite{HopKar,MicVaz,Vaz}), the same is true for $G$.
\end{proof}

Since in a forest, every $P_4$ is an induced subgraph, we have the following positive result 
complementing Theorem~\ref{thm:apxlf}:
 
\begin{corollary}
In every $P_4$-free forest $F$, \mrbm\ is solvable in time 
$O(n^{3/2})$, where $n$ is the number of vertices in $F$.
\end{corollary}

In contrast to Theorem~\ref{thm:apxtree}, saying that \mrbm\ is hard even for $P_8$-free 
trees, we have:

\begin{theorem}\label{thm:stree}
In every tree $T$ which does not contain $P_7$ as a subgraph, \mrbm\ is solvable in time 
$O(n^{7/2})$, where $n$ is the number of vertices in $T$.
\end{theorem}
\begin{proof}
As $T$ is $P_7$-free, we can find an edge $xy$ in $T$ such that $G-\{x,y\}$ is a forest 
consisting of stars; all we have to do is to pick the two most central vertices in a 
longest path in $T$ and note that every vertex of $T$ must have distance at most $2$ 
to $\{ x,y\}$. Every matching $M$ can contain at most $2$ edges incident to $\{x,y\}$. 
Once we have decided on these at most two edges (less than $n^2$ choices), we are left 
with the task of finding a rainbow matching in a $P_4$-free forest, which can be done 
in time $O(n^{3/2})$ by Theorem~\ref{thm:sfor}. This gives a total time of $O(n^{7/2})$.
\end{proof}

The following theorem describes a more general setting of Theorem~\ref{thm:stree}:

\begin{theorem}\label{thm:sfor7}
In every forest $F$ which does not contain $P_7$ as a subgraph, \mrbm\ is solvable in time 
$O\left(\frac{1}{2^kk^{2k}}n^{(4k+3)/2}\right)$, where $n$ is the number of vertices in 
$F$ and $k$ is the number of components in $F$.
\end{theorem}
\begin{proof}
 As in proof of Theorem~\ref{thm:stree}, find a central edge $x_Ty_T$ in every component 
$T\subseteq F$, such that $\bigcup\big(T-\{x_T,y_T\}\big)$ is a forest consisting of stars. 
Once we have decided on the at most $2k$ edges in a matching incident to the $x_Ty_T$---a 
total of at most ${n/k\choose 2}^k<\frac{n^{2k}}{2^kk^{2k}}$ choices---we are left with 
the task of finding a rainbow matching in a $P_4$-free forest, which can be done in 
$O(n^{3/2})$ by Theorem~\ref{thm:sfor}. This gives a total time of 
$O\left(\frac{1}{2^kk^{2k}}n^{(4k+3)/2}\right)$.
\end{proof}

\begin{corollary}
\mrbm\ is solvable in polynomial time for $P_7$-free forests with bounded number of components.
\end{corollary}

\section{Fixed Parameter Aspects}\label{sec:fixed}

An approach to deal with NP-hard problems is to fix a parameter when solving the problems. 
A problem parameterized by $k$ is {\it fixed parameter tractable}, {\it fpt} for short, 
if it can be solved in time $f(k)\cdot n^{O(1)}$, or, equivalently, in time 
$O\left(n^{O(1)}+f(k)\right)$, where $f(k)$ is a computable function, depending only on 
the parameter $k$. 
For an introduction to parameterized complexity theory, see for instance~\cite{DF,FG,Nie}.

Observe that \mrbm\ is fpt, when parameterized by the size of the problem solution. 
In case of properly edge-colored inputs, this is a consequence of Lemma~\ref{lemma:CL}, 
and of the fact that {\sc MIS} is fpt for $(K_7-e)$-free graphs (\cite{DLMR}). 
In the general case, rainbow matchings can be seen as matching (set packing) in certain 
$3$-uniform hypergraphs, hence \mrbm\ is fpt by a result of Fellows et~al.~\cite{FKNRRSTW}.

Recall that \mrbm\ is already hard for $P_8$-free trees. In view of Theorem~\ref{thm:sfor7}, 
we now consider \mrbm\ for $P_7$-free forests, parametri\-zed by $k$, the number of components 
in the inputs. Formally, we want to address the following parameterized problem:

\medskip
\noindent
$k$-{\sc forest rainbow matching}\\[1ex] 
\begin{tabular}{lp{0.8\textwidth}}
Instance: & A $P_7$-free forest $F$ with $k$ components containing a $P_4$.\\
Parameter: & $k$.\\
Output: & A maximum rainbow matching in $F$.
\end{tabular}

\bigskip
\noindent
Theorem~\ref{thm:sfor5} below shows that $k$-{\sc forest rainbow matching} is fpt for $P_5$-free 
forests.

\begin{theorem}\label{thm:sfor5}
In every forest $F$ which does not contain $P_5$ as a subgraph, \mrbm\ is solvable in time 
$O(n+2^kk^{3})$, where $n$ is the number of vertices in $F$ and 
$k$ is the number of components in $F$ containing a $P_4$.
\end{theorem}
\begin{proof}
Let $E'\subset E(F)$ be the set of edges between vertices of degree greater than 
$1$, and observe that $|E'|\le k$. Then, $F'=F-E'$ is a forest consisting of at 
most $2k$ stars.
Delete edges in $F'$ until each star is rainbow colored and has at most $2k$ edges. 
Notice that this does not change the size of a maximum rainbow matching. Call this 
new graph $F''$, and observe that $|E(F'')|\le 4k^2$. 

Now for every rainbow choice of edges in $E'$, we can solve \mrbm\ on a subgraph 
of $F''$. There are at most $2^k$ rainbow choices in $E'$, and solving \mrbm\ on 
$F''$ takes time $O(k^3)$ as in Theorem~\ref{thm:sfor}. Creating $F''$ from $F$ 
takes time $O(n)$, so the result follows.
\end{proof}

We do not know if $k$-{\sc forest rainbow matching} is fpt for $P_6$-free forests or 
$W[1]$-hard.  
Note that $k$-{\sc forest rainbow matching} is in XP by Theorem~\ref{thm:sfor7}. 

\begin{remark}\label{thm:sfor67}
If $k$-{\sc forest rainbow matching} for $P_6$-free forests is in $W[i]$, then 
so is $k$-{\sc forest rainbow matching} for $P_7$-free forests.
\end{remark}
\begin{proof}
Let $F$ be a $P_7$-free forest with $k$ components containing a $P_4$.
 For every tree $T\subseteq F$ containing a $P_6$, find the central edge $x_Ty_T$. 
Once you choose one of the $2^k$ possibilities to include these central edges in a 
rainbow matching, delete the not-chosen edges, and delete the chosen edges together 
with their neighborhood, the remaining graph contains at most $2k$ components 
containing a $P_4$, and no components containing a $P_6$.
\end{proof}

\section{Concluding Remarks}\label{sec:conclusion}
We have shown that it is NP-hard to approximate \mrbm\ within certain ratio bounds 
for very restricted graph classes. Implicit in our results is the following dichotomy 
theorem for forests and trees in terms of forbidding paths: \rbm\ is NP-complete for 
$P_5$-free forests ($P_8$-free trees), and is polynomially solvable for $P_4$-free forests ($P_7$-free trees). 

We have also proved that $k$-{\sc forest rainbow matching} is fixed parameter tracta\-ble 
for $P_5$-free forests. What can we find out about the parameterized complexity in the only open case of $P_6$-free 
(and equivalently, $P_7$-free) forests? 

Another open problem of independent interest is the computational complexity of recognizing color-line graphs: Given a graph $G$, does there exist an edge-colored graph 
$H$ such that $G=C\!L(H)$? 
Note that it is well-known that line graphs can be recognized in linear time.

\bibliographystyle{amsplain}

\begin{thebibliography}{99}

\bibitem{AK} 
  Paola Alimonti and Viggo Kann,
  Some APX-completeness results for cubic graphs,
  {\it Theoretical Computer Science} 237 (2000) 123--134.

\bibitem{BK} 
  Piotr Berman and Marek Karpinski, 
  On Some Tighter Inapproximability Results,
  ICALP'99, {\it Lecture Notes in Comput. Sci.,} 1644 (1999), pp. 200--209.

\bibitem{DLMR}
  Konrad Dabrowski, Vadim V. Lozin, Haiko M\"uller, and Dieter Rautenbach, 
  Parameterized algorithms for the independent set problem in hereditary graph classes, 
  {\it J. of Discrete Algorithms} 14 (2012) 207--213. 
   
\bibitem{DFLMPW} 
  Jennifer Diemunsch, Michael Ferrara, Allan Lo, Casey Moffatt, Florian Pfender, 
  and Paul S. Wenger, 
  Rainbow matchings of size $\delta(G)$ in properly-colored graphs, 
  {\it Electr. J. Combin.} 19(2) (2012), \#P52.

\bibitem{DF}
  Rodney G. Downey and Michael R. Fellows,
  {\it Parameterized Complexity}, 
  Springer-Verlag, New York, 1999.
  
\bibitem{FKNRRSTW}
  Michael R. Fellows, Christian Knauer, Naomi Nishimura, Prabhakar Ragde, 
  Frances A. Rosamond, Ulrike Stege, Dimitrios M. Thilikos, Sue Whitesides, 
  Faster Fixed-Parameter Tractable Algorithms for Matching and Packing Problems, 
  {\it Algorithmica} 52 (2008) 167--176.

\bibitem{FG}
  J\"org Flum and Martin Grohe, 
  {\it Parameterized Complexity Theory},
  Springer-Verlag, Berlin Heidelberg, 2006.

\bibitem{GJ} 
  Michael R. Garey and David S. Johnson, 
 \emph{Computers and Intractability: An Introduction to the Theory of NP-completeness}, 
  Freeman, New York, 1979.

  

\bibitem{HopKar}
  John E. Hopcroft and Richard M. Karp, 
  A $n^{5/2}$ algorithm for maximum matching in bipartite graphs,
  {\it SIAM J. on Comput.} 2 (1973) 225--231.

\bibitem{Hal} Magn\'{u}s M. Halld\'{o}rsson,
  Approximations of independent sets in graphs. Approximation algorithms for combinatorial 
  optimization (Aalborg, 1998), 1--13,
  {\it Lecture Notes in Comput. Sci.,} 1444, Springer, Berlin, 1998. 

\bibitem{HS} 
  Cor A. J. Hurkens and Alexander Schrijver, 
  On the size of systems of sets every $t$ of which have an SDR, with an application to 
  the worst-case ratio of heuristics for packing problems, 
  {\em SIAM Journal on Discrete Mathematics} 2 (1989) 68--72.

\bibitem{IRT}
  Alon Itai, Michael Rodeh, and Steven L. Tanimoto,
  Some Matching Problems for Bipartite Graphs, 
  \emph{Journal of the Association for Computing Machinery}, 25 (1978) 517--525.

\bibitem{KL}
  Mikio Kano and Xueliang Li, 
  Monochromatic and Heterochromatic Subgraphs in Edge-Colored Graphs -- A Survey,
  \emph{Graphs and Combinatorics} 24 (2008) 237--263.
  
\bibitem{KY}
  Alexandr V. Kostochka and M. Yancey, 
  Large rainbow matchings in edge-coloured graphs,  
  \emph{Combinatorics, Probability and Computing} 21 (2012) 255--263.

\bibitem{LT} 
  Allan Lo and Ta Sheng Tan, 
  A note on large rainbow matchings in edge-coloured graphs,
  to appear in {\it Graphs and Combinatorics}. DOI: 10.1007/s00373-012-1271-y


\bibitem{MicVaz}
  Silvio Micali and Vijay V. Vazirani, 
  An $O(\sqrt{V}E)$ Algorithm for Finding Maximum Matching in General Graphs, 
  in: {\it Proc. 21st Annual IEEE Symposium on Foundations of Computer Science} (1980) 17--27. 

\bibitem{Nie}
  Rolf Niedermeier, 
  {\it Invitation to Fixed-Parameter Algorithms},
  Oxford University Press, 2006.

\bibitem{PY}
  Christos H. Papadimitriou and Mihalis Yannakakis, 
  Optimization, approximation, and complexity classes, 
  {\it J. Comput. Syst. Sci.} 43 (1991) 425--440.
  
\bibitem{Pet} Julius Petersen, 
  Die Theorie der regul\"aren Graphen,
  {\em Acta Math.} 15 (1891), 193--220. 

\bibitem{Ry} 
  Herbert J. Ryser, 
  Neuere Probleme der Kombinatorik,
  {\em Vortr\"age \"uber Kombinatorik},  
  Ma\-the\-ma\-ti\-sches Forschungsinstitut Oberwolfach, Oberwolfach, 1967, 24--29.

\bibitem{Vaz}
  Vijay V. Vazirani, 
  A Theory of Alternating Paths and Blossoms for Proving Correctness of the 
  $O(\sqrt{V} E)$ Maximum Matching Algorithm,
  {\it Combinatorica} 14 (1994) 71--109.
  
\bibitem{Wa} Guanghui Wang,
  Rainbow matchings in properly edge colored graphs,
  {\em Electr. J. Combin.} 18(1) (2011), \#P162.

\end{thebibliography}

\end{document}